\documentclass{article}

\usepackage{tikz}
\usepackage{amsfonts}
\usepackage{amssymb,epic,eepic}
\usepackage{amsmath}
\usepackage{dcolumn}
\usepackage{bm}
\usepackage{bbm}
\usepackage{verbatim}
\usepackage{color}
\usepackage{amsthm}

\newcommand{\hr}{{\mathcal H}}
\newcommand{\cn}{{\mathcal N }}

\newcommand{\cs}{{\mathcal S}}
\newcommand{\crr}{{\mathcal R}}
\newcommand{\fr}{{\mathcal F}}

\newcommand{\fri}{{\mathfrak I}}
\newcommand{\kr}{{\mathcal K}}

\newcommand{\cc}{{\mathbb C}}
\newcommand{\rr}{{\mathbb R}}
\newcommand{\M}{{\mathcal M}}
\newcommand{\nn}{{\mathbb N}}

\newcommand{\eps}{{\varepsilon}}        

\newcommand{\A}{\mathcal A}

\newcommand{\D}{\mathcal D}
\newcommand{\cP}{\mathcal P}
\newcommand{\bS}{\mathbf S}
\newcommand{\bX}{\mathbf X}

\newcommand{\eins}{{\mathbbm{1}}}

\newtheorem{theorem}{Theorem}

\newtheorem{corollary}{Corollary}

\newtheorem{definition}{Definition}

\newtheorem{lemma}{Lemma}

\newtheorem{remark}{Remark}

\newcommand{\tr}{\mathrm{tr}}

\DeclareMathOperator{\conv}{conv}

\DeclareMathOperator{\linspan}{span}

\begin{document}

\title{Positivity, Discontinuity, Finite Resources and Nonzero Error for Arbitrarily Varying Quantum Channels}
\author{H. Boche, J. N\"otzel \\
\scriptsize{Electronic addresses: \{boche, janis.noetzel\}@tum.de}
\vspace{0.2cm}\\
{\footnotesize Lehrstuhl f\"ur Theoretische Informationstechnik, Technische Universit\"at M\"unchen,}\\
{\footnotesize 80290 M\"unchen, Germany}
}
\maketitle

\begin{abstract}
This work is motivated by a quite general question: Under which circumstances are the capacities of information transmission systems continuous? The research is explicitly carried out on arbitrarily varying quantum channels (AVQCs).\\
We give an explicit example that answers the recent question whether the transmission of messages over AVQCs can benefit from distribution of randomness between the legitimate sender and receiver in the affirmative.\\
The specific class of channels introduced in that example is then extended to show that the deterministic capacity does have discontinuity points, while that behaviour is, at the same time, not generic: We show that it is continuous around its positivity points. This is in stark contrast to the randomness-assisted capacity, which is always continuous in the channel. Our results imply that the deterministic message transmission capacity of an AVQC can be discontinuous only in points where it is zero, while the randomness assisted capacity is nonzero.\\
Apart from the zero-error capacities, this is the first result that shows a discontinuity of a capacity for a large class of quantum channels. The continuity of the respective capacity for memoryless quantum channels had, among others, been listed as an open problem on the problem page \cite{problempage} of the ITP Hannover for about six years before it was proven to be continuous.\\
We also quantify the interplay between the distribution of finite amounts of randomness between the legitimate sender and receiver, the (nonzero) decoding error with respect to the average error criterion that can be achieved over a finite number of channel uses and the number of messages that can be sent. This part of our results also applies to entanglement- and strong subspace transmission.\\
In addition, we give a new sufficient criterion for the entanglement transmission capacity with randomness assistance to vanish.
\end{abstract}
\begin{section}{Introduction}
A key property of a communication system is a continuous dependence of its performance on the system parameters. If small perturbations of the system lead to dramatic losses in the performance, it will most likely not be too widely used and instead be replaced by a more robust system.\\
The most fundamental property of a communication system is its capacity, and the very definition of capacities in quantum Shannon information theory is generically such that a straightforward application of Fekete's Lemma proves their existence as a real number.\\
One could now raise the question whether a similarly straightforward method exists that enables one to prove the continuity of channel capacities. The known methods are built upon the continuity of the entropy and require an explicit entropic formula for the capacity. Using them, one can prove \cite{leung-smith} that the capacities of memoryless quantum channels are continuous, and these methods carry over to the randomness-assisted capacities of AVQCs \cite{abbn}. At the same time, the zero-error capacities (see e.g. \cite{duan-severini-winter}) are certainly not continuous (this has e.g. been proven in \cite{abbn}), and there exists no explicit formula for computing them.\\
At this point, we observe that there is no known formula for the deterministic message transmission capacity of an AVQC, so we may rightfully ask whether it is continuous or not. Our investigation revealed that there do indeed exist discontinuity points for the unassisted (deterministic) message transmission capacity of an AVQC. This finding is intimately connected to the difference between randomness-assisted and unassisted message transmission capacity of AVQCs, a difference which has explicitly been conjectured in \cite{bn-correlation} and is proven in this paper.\\
The question whether a general method for proving continuity of a communication system exists remains open.\\
As an additional result we give bounds on the functional dependence between blocklength, error and the amount of common randomness needed to achieve that error.
\\\\
In the following, we will explain the model of an arbitrarily varying channel and provide examples for communication scenarios whose essential features are captured by the model. We then explain the effect of shared randomness for these systems and state a corresponding result. In close connection, we discuss the relevance of continuity of capacities, state results and give examples. Finally, we quantify the interplay between finite errors, block length and the amount of common randomness needed to achieve that error.
\\\\
Imagine a situation where a sender wants to transmit for example messages to a remote receiver. They each have access to a quantum system which is modeled on a finite dimensional Hilbert space and are connected by a quantum channel. Dependent on the message he wants to transmit the sender prepares some quantum state, which is then transmitted to the receiver over the channel. The question then is, whether the receiver can infer which message the sender intended to send just by performing measurements on the output states.\\
We assume that multiple channel uses are available and that the channel does not have a memory - but instead assume the existence of a jammer, which tries to prevent the two legal parties from communicating properly. Such a situation can arise e.g. in secret key distribution or transmission scenarios over quantum channels as developed by Devetak in \cite{devetak}, but when the evil third party is either not interested in or unable to do eavesdropping on the legal communication, but has some influence on the channel between the legal parties. The power of the jammer is, in the model chosen here, precisely quantified by his ability to influence the channel:\\
He is able to choose, for each of the multiple channel uses, one out of a fixed set $\fri$ of channels. This set is known to all three parties. The goal of sender and receiver is now to find encoding-and decoding procedures such that they can reliably transmit their data, no matter which choice the jammer makes. It can even be assumed that the jammer knows \emph{in advance} how the encoding-decoding procedure of sender and receiver works. This assumption will always be satisfied in commercial communication systems, where standardized protocols are being used. The model that we just introduced is called an arbitrarily varying quantum channel. Note that, throughout the entire manuscript, we restrict attention to finite AVQCs, e.g. those for which $|\fri|<\infty$ holds. The main reason for this is that it greatly simplifies proofs and puts a clean focus on the most relevant features of the systems under consideration.\\
Of course, the very same model can be formulated by using as the basic channels either classical, classical-quantum or quantum-classical channels, and the underlying systems that the three parties act upon could be described by any kind of physical theory. Another possible change in the model would be to enable the jammer to use quantum inputs to the system. In this work, we will stick to the model we described first.\\
The situation described by the model can, in these days, be found in denial-of-service attacks. It is important to note that the quality of an arbitrarily varying classical channel can not only be described by entropic quantities, as is the case for stationary memoryless channels. It has rather been found that its capability to transmit any messages at all is completely characterized by so-called \emph{symmetrizability conditions}.\\
Let us get into a bit more detail here. It has been proven, first in \cite{ahlswede-elimination} for classical arbitrarily varying channels, then in \cite{ahlswede-blinovsky} for classical-quantum arbitrarily varying channels that these systems exhibit a dichotomic behaviour: the message-transmission capacity under average error criterion, $\overline C_{\mathrm{det}}$, is either zero or equals an easily computable number, called the random capacity $\overline C_{\mathrm{random}}$. The latter quantity is the amount of messages that can be sent with transmission error approaching zero, when the number of channel uses goes to infinity and sender and receiver share a sufficiently large amount of shared randomness (polynomially much common randomness, in the number of channel uses, is sufficient). It turned out later \cite{csiszar-narayan,ericson} that those arbitrarily varying channels $\mathfrak W$ for which $\overline C_{\mathrm{det}}(\mathfrak W)=0$ holds are exactly characterized by so-called ``symmetrizability'' conditions.\\
The dichotomic behaviour has been proven to hold true for both entanglement and message transmission over AVQCs in \cite{abbn}. Another result of the work \cite{abbn} was that encoding-decoding schemes for entanglement transmission are also good for strong subspace transmission and vice versa. The later work \cite{bn-correlation} showed that this is also true for message transmission under average- and maximal error criterion. These results enable us to restrict our discussion to the average error criterion and entanglement transmission henceforth.\\
Despite these achievements, it remained an open question until now whether shared randomness really helps the transmission of messages over AVQCs, and the same question remained open for entanglement-and strong subspace transmission.\\
More precisely, it has been conjectured in \cite{abbn} that shared randomness does not increase the entanglement transmission capacity of AVQCs and in \cite{bn-correlation} that there exist examples of AVQCs $\fri$ for which $\overline C_{\mathrm{random}}(\fri)>0$ but $\overline C_{\mathrm{det}}(\fri)=0$ holds.\\
In this work we provide exactly such an example.\\
We then study the continuity properties of $\overline C_{\mathrm{det}}$ for AVQCs. We find that $\overline C_{\mathrm{det}}$ is continuous around every AVQC $\fri$ for which $\overline C_{\mathrm{det}}(\fri)>0$ holds. Put into simple words: If a system which is modeled as an AVQC is 'useful' in the sense that $\overline C_{\mathrm{det}}(\fri)>0$, then this remains true even if small errors are present in the evaluation of the system parameters.\\
An obvious question that comes with the above two results is, whether there really exist discontinuities for the function $\fri\mapsto\overline C_{\mathrm{det}}(\fri)$. The continuity of the message- and entanglement transmission capacity of a stationary memoryless quantum channel has been an open problem for quite a while, it was posed by M. Keyl and listed in the open problem page \cite{problempage} of R. Werner's group since 2003. After partial results, it was completely solved by Leung and Smith in \cite{leung-smith} in 2009, and answered in the affirmative: Both message- and entanglement transmission capacity are continuous for stationary memoryless quantum channels.\\
Quite on the contrary, we prove in this work that the message transmission capacity of AVQCs \emph{without} assistance by shared randomness is not continuous. We do so by explicit construction of an example. This is the first example of a discontinuous behaviour of a quantum capacity other than the zero-error capacities \cite{duan-severini-winter}.\\
Our previous results clearly demonstrate the importance of shared randomness for AVQCs. In \cite{abbn}, R. Ahlswede, I. Bjelakovi\c{c} and the authors showed that already a small amount of common randomness is sufficient to ensure that transmission of messages is possible at rates arbitrarily close to $\overline C_{\mathrm{random}}$. The same holds true for transmission of entanglement. Building on that and the work \cite{ahlswede-cai} of R. Ahlswede and N. Cai, the authors were able to to show in \cite{bn-correlation} that already the use of arbitrarily small amounts of correlation yield the same result.\\
This demonstrates that shared randomness has two important effects for AVQCs: First, it boosts the capacity to the maximally possible value, and second it stabilizes the system with respect to small changes (the capacity function with assistance by either unlimited shared randomness, positive correlation or small amounts of common randomness is always continuous).\\
This gives a strong motivation to start a closer investigation of the exact interplay between the system parameters, the error of message transmission at a specific block length and the amount of randomness used for stabilization of the system. This investigation is carried out in the last part of the paper. We give bounds on the number of shared secret bits (common randomness) $K$ needed to achieve some pre-given maximal error $\lambda$ within $L$ channel uses. Assuming that the AVQC under consideration has $|\bS|$ constituents, the scaling law is roughly $ K\leq\frac{\log|\bS|}{E\cdot\lambda}$, where $E$ is the reliability function of the \emph{compound channel} $\conv(\fri)$ and $L$ is more implicitly given through $E$, roughly scaling as $L(1-\frac{1}{L}\log L)\approx -E\log(E/\lambda)$. In case that the AVQC $\fri$ is symmetrizable, we note that the results of \cite{abbn} imply $\frac{1}{2\lambda}\leq K$, and for non-symmetrizable AVQCs we know that $K=0$ is sufficient by the quantum-Ahlswede dichotomy proven in \cite{abbn}.\\
Another important observation is that the number $K$ of random bits needed to guarantee a certain quality of transmission is essentially independent of the number $l$ of channel uses, if only $l\geq L$ holds, and is indefinite for $l< L$.\\
The technique of proof we utilize here applies to entanglement transmission as well.\\
It is clear that a similar result could be obtained by using only correlation to first establish enough common randomness and then use it with the above stated bounds, but the exact trade-off between $\lambda$, $L$ and the 'amount' of correlation remains unclear and we leave that question open for future work.
\\\\
A historical overview concerning the history of arbitrarily varying channels (in both the classical and the quantum case) can be found in \cite{abbn}. Among the more recent developments are \cite{abbn}, \cite{bbs}, \cite{wb} and \cite{bn-correlation}.
\end{section}
\begin{section}{\label{sec:Notation}Notation}
All Hilbert spaces are assumed to have finite dimension and are over the field $\cc$. The set of linear operators from $\hr$ to $\hr$ is denoted $\mathcal B(\hr)$. The adjoint of $b\in\mathcal B(\hr)$ is marked by a star and written $b^\ast$.\\
$\cs(\hr)$ is the set of states, i.e. positive semi-definite operators with trace (the trace function on $\mathcal B(\hr)$ is written $\tr$) $1$ acting on the Hilbert space $\hr$. The maximally mixed state with only one eigenvalue $\dim(\hr)$ in $\cs(\hr)$ is written $\pi_\hr$ or, if no confusion can arise, simply $\pi$. Pure states are given by projections onto one-dimensional subspaces. A vector $x\in\hr$ of unit length spanning such a subspace will therefore be referred to as a state vector, the corresponding state will be written
$|x\rangle\langle x|$. For a finite set $\mathbf X$ the notation $\mathfrak{P}(\mathbf X)$ is reserved for the set of probability distributions on $\mathbf X$, and
$|\mathbf X|$ denotes its cardinality. For any $l\in\nn$, we define $\bX^l:=\{(x_1,\ldots,x_l):x_i\in\bX\ \forall i\in\{1,\ldots,l\}\}$, we also write $x^l$ for the elements of $\bX^l$.\\
The set of completely positive trace preserving (CPTP) maps (also called quantum channels)
between the operator spaces $\mathcal{B}(\hr)$ and $\mathcal{B}(\kr)$ is denoted by $\mathcal{C}(\hr,\kr)$.\\
Closely related is the set of classical-quantum channels (abbreviated here using the term 'cq-channels') with finite input alphabet $\mathbf Z$ and output alphabet $\kr$, that arises from $\mathcal C(\hr,\kr)$ by setting $d=|\mathbf Z|$ and restricting the inputs to matrices that are diagonal in any specific basis. This set is denoted $CQ(\mathbf Z,\kr)$. Both $\mathcal C(\hr,\kr)$ and $CQ(\mathbf Z,\kr)$ are convex subsets of vector spaces.\\
For any natural number $N$, we define $[N]$ to be the shorthand for the set $\{1,...,N\}$.\\
Using the usual operator ordering symbols $\leq$ and $\geq$ on $\mathcal B(\hr)$, the set of measurements with $N\in\nn$ different outcomes is written
\begin{align}\M_N(\hr):=\{\mathbf D:\mathbf D=(D_1,\ldots,D_N)\ \wedge\ \sum_{i=1}^ND_i\leq\eins_\hr\ \wedge\ D_i\geq0\ \forall i\in[N]\}.\end{align}
To every $\mathbf D\in \M_N(\hr)$ there corresponds a unique operator defined by $D_0:=\eins_\hr-\sum_{i=1}^ND_i$. Throughout the paper, we will assume that $D_0=0$ holds. This is possible in our scenario, since adding the element $D_0$ to any of the other $D_1,\ldots,D_N$ does not decrease the performance of a given code.\\
The von Neumann entropy of a state $\rho\in\mathcal{S}(\hr)$ is given by
\begin{equation}S(\rho):=-\textrm{tr}(\rho \log\rho),\end{equation}
where $\log(\cdot)$ denotes the base two logarithm which is used throughout the paper.\\
The Holevo information is for a given channel $W \in CQ(\mathbf{X},\hr)$ and input probability distribution $p \in \mathfrak P(\mathbf{X})$ defined by
\begin{align}
 \chi(p, W) := S(\overline{W}) - \sum_{x \in \mathbf{X}} p(x) S(W(x)),
\end{align}
where $\overline{W}$ is defined by $\overline{W}:= \sum_{x \in \mathbf{X}} p(x) W(x)$\\
Given a bipartite random variable $(X,Y)$, its mutual information $I(X,Y)$ is given by $I(X,Y):=H(X)+H(Y)-H(X,Y)$, where $H(\cdot)$ is the usual Shannon entropy. If $p\in\mathfrak P(\bX)$ for some set $\bX$ satisfying $|\bX|=2$, its Shannon entropy is also written as $h(p(x)):=H(p)$, where $x\in\bX$ is arbitrary.\\
For $\rho\in\mathcal{S}(\hr)$ and $\cn\in
\mathcal{C}(\hr,\hr)$ the entanglement fidelity (which was defined in \cite{schumacher}) is given by
\begin{equation}F_e(\rho,\cn):=\langle\psi, (id_{\mathcal{B}(\hr)}\otimes \cn)(|\psi\rangle\langle \psi|)     \psi\rangle,  \end{equation}
with $\psi\in\hr\otimes \hr$ being an arbitrary purification of the state $\rho$.\\
For a finite set $\mathcal{W}=\{W_s\}_{s\in\bS} \subset \mathcal C(\hr,\kr)$ or $\mathcal{W}=\{W_s\}_{s\in\bS} \subset CQ(\mathbf Z,\kr)$ we denote its convex hull by $\conv(\mathcal{W})$. In the cases considered here the following will be sufficient. For a set $\mathcal{W}:= \{W_s\}_{s \in \bS}$
\begin{align}\label{eq:conv-hull}
 \conv(\mathcal{W})=\left\{W_{q}: W_q=\sum_{s\in \bS}q(s)W_s,\ q \in\mathfrak{P}(\bS)\right\}.
\end{align}
As a measure of distance between channels we use the diamond norm $||\cdot||_\lozenge$, which is given by
\[||\cn||_{\lozenge}:=\sup_{n\in \nn}\max_{a\in \mathcal{B}(\cc^n\otimes\hr),||a||_1=1}||(\textrm{id}_{n}\otimes \mathcal{N})(a)||_1,   \]
where $\textrm{id}_n:\mathcal{B}(\cc^n)\to \mathcal{B}(\cc^n)$ is the identity channel, and $\mathcal{N}:\mathcal{B}(\hr)\to \mathcal{B}(\kr)$ is any linear map, not necessarily completely positive. The merits of $||\cdot||_{\lozenge}$ are due to the following facts (cf. \cite{kitaev}). First, $||\cn||_{\lozenge}=1$ for all $\cn\in\mathcal{C}(\hr,\kr)$. Thus, $\mathcal{C}(\hr,\kr)\subset S_{\lozenge}$, where $S_{\lozenge}$ denotes the unit sphere of the normed space $(\mathcal{B}(\mathcal{B}(\hr),\mathcal{B}(\kr)),||\cdot||_{\lozenge} )$. Moreover, $||\cn_1\otimes \cn_2||_{\lozenge}=||\cn_1||_{\lozenge}||\cn_2||_{\lozenge}$ for arbitrary linear maps $\cn_1,\cn_2:\mathcal{B}(\hr)\to \mathcal{B}(\kr) $.\\
We further use the diamond norm to define the function $D_\lozenge(\cdot,\cdot)$ on $\{(\fri,\fri'):\fri,\fri'\subset\mathcal C(\hr,\kr)\}$, which is for $\fri,\fri'\subset\mathcal C(\hr,\kr)$ given by
$$D_\lozenge(\fri,\fri'):=\max\{\sup_{\cn\in\fri}\inf_{\cn'\in\fri'}||\cn-\cn'||_\lozenge,\sup_{\cn'\in\fri'}\inf_{\cn\in\fri}||\cn-\cn'||_\lozenge\}.$$
For $\fri\subset\mathcal C(\hr,\kr)$ let $\bar\fri$ denote the closure of $\fri$ in $||\cdot||_\lozenge$. Then $D_\lozenge$ defines a metric on $\{(\fri,\fri'):\fri,\fri'\subset\mathcal C(\hr,\kr),\ \fri=\bar\fri,\ \fri'=\bar{\fri'}\}$ which is basically the Hausdorff distance induced by the diamond norm.
\end{section}
\begin{section}{\label{sec:Definitions}Definitions}
For the rest of this subsection, let $\fri=\{\cn_s\}_{s\in\bS}\subset\mathcal C(\hr,\kr)$ denote a finite set of channels and $\hr,\kr$ some arbitrary but fixed finite dimensional Hilbert spaces over $\mathbb C$. Henceforth, we follow the convention from \cite{abbn}, using the term 'the AVQC $\fri$' as a linguistic shorthand for the mathematical object $(\{\cn_{s^l}\}_{s^l\in\bS^l})_{l\in\nn}$.\\
Due to the close correspondence between arbitrarily varying and certain compound channels, we will sometimes also encounter the case that $\fri$ stands for the compound channel $(\{\cn_{q}^{\otimes l}\}_{\cn_q\in\conv(\fri)})_{l\in\nn}$. In those cases, this will be explicitly mentioned.\underline{}
We will now define the entanglement transmission capacities of an AVQC. Corresponding coding theorems can be found in \cite{abbn}.
\begin{definition}
An $(l,k_l)-$\emph{random entanglement transmission code} for $\fri$ is a probability measure $\mu_l$ on $(\mathcal C(\fr_l,\hr^{\otimes l})\times\mathcal C(\kr^{\otimes l},\fr_l'),\sigma_l)$,
where $\fr_l,\fr_l'$ are Hilbert spaces, $\dim\fr_l=k_l$, $\fr_l\subset\fr_l'$ and the sigma-algebra $\sigma_l$ is chosen such that the function $(\cP_l,\crr_l)\mapsto F_e(\pi_{\fr_l},\crr_l\circ\cn_{s^l}\circ\cP_l)$ is measurable
w.r.t. $\sigma_l$ for every $s^l\in\bS^l$.\\
Moreover, we assume that $\sigma_l$ contains all singleton sets. An example of such a sigma-algebra $\sigma_l$ is given by
the product of sigma-algebras of Borel sets induced on $\mathcal C(\fr_l,\hr) $ and $\mathcal C(\kr,\fr_l') $ by the standard topologies of the ambient spaces.\\
The error of the code is given by $\eps_l:=1-\int d\mu_l(\crr^l,\cP^l)F_e(\pi_{\fr_l},\crr_l\circ\cn_{s^l}\circ\cP_l)$.
\end{definition}
\begin{definition}\label{def:random-cap-ent-trans}
A non-negative number $R$ is said to be an achievable entanglement transmission rate for the AVQC $\fri=\{\cn_s  \}_{s\in\bS}$ with random codes and error $\lambda\in[0,1]$ if there is a sequence of $(l,k_l)-$random entanglement transmission codes such that
\begin{enumerate}
\item $\liminf_{l\rightarrow\infty}\frac{1}{l}\log k_l\geq R$ and
\item $\liminf_{l\rightarrow\infty}\inf_{s^l\in\bS^l}\int F_e(\pi_{\fr_l},\crr^l\circ\cn_{s^l}\circ\cP^l)d\mu_l(\cP^l,\crr^l)\geq1-\lambda$.
\end{enumerate}
The random entanglement transmission capacity $\A_{\textup{random}}(\fri,\lambda)$ of $\fri$ with error $\lambda$ is defined by
\begin{equation}
\A_{\textup{random}}(\fri,\lambda):=\sup\left\{R\in\rr_+:\begin{array}{l}R \textrm{ is an achievable entanglement trans-}\\ \textrm{mission rate for } \fri \textrm{ with random codes and error }\lambda\end{array}\right\}.
\end{equation}
\end{definition}
\begin{remark}
The reader with some previous knowledge about the topic of finite errors will notice that this definition differs from the classical one used e.g. in \cite{ahlswede-wolfowitz}. Precisely speaking, if one would define capacities with finite errors in the spirit of \cite{ahlswede-wolfowitz} and use the symbol $\tilde\A_{\mathrm{random}}$ for those, then one would set $\tilde\A_{\mathrm{random}}:=\lim_{n\to\infty}\frac{1}{n}\log\max\{k_l:\exists\ (l,k_l)-\mathrm{code\ for\ entanglement\ transmission\ with\ }\lambda_l\leq\lambda\}$. Since $\lambda\to\tilde\A_{\mathrm{random}}(\fri,\lambda)$ is monotone increasing on $[0,1]$, the limits $\lim_{\eps\to0}\tilde\A_{\mathrm{random}}(\fri,\lambda+\eps)$ exist for every $\lambda\in[0,1)$. It is then clear that the equality $\A_{\mathrm{random}}(\fri,\lambda)=\lim_{\eps\to0}\tilde\A_{\mathrm{random}}(\fri,\lambda+\eps)$ holds for all $\lambda\in[0,1)$. We thus see that $\A_{\mathrm{random}}$ is simply the right-regularized version of $\tilde\A_{\mathrm{random}}$.\\
While $\tilde A_{\mathrm{random}}$ might be a practically more relevant definition, it is clear that the two definitions can lead to a different value in capacity only at discontinuity points of $\tilde A_{\mathrm{random}}$. Since both functions are monotone increasing on the interval $[0,1]$, the number of such points is countable by \cite{rudin-principles-of-mathematical-analysis}, Theorem 4.30.\\
Notably, at $\lambda=0$, one gets $\A_{\mathrm{random}}(\fri,0)=\A_{\mathrm{random}}(\fri)$ for 'the' random capacity $\A_{\mathrm{random}}$ of an AVQC according to Definition 2 in \cite{abbn}, while $\tilde\A_{\mathrm{random}}(\fri,0)$ gives the randomness-assisted zero-error capacity of an AVQC.\\
This latter point makes our definition fit seamlessly with the previous work \cite{abbn,bn-correlation} on AVQCs. At the same time, we do not encounter a dramatically different behaviour in most cases. The same reasoning applies to all the other capacities defined in this paper.
\end{remark}
Having defined random codes and random code capacity for entanglement transmission we are in the position to introduce their deterministic counterparts: An $(l,k_l)-$code for entanglement transmission over $\fri$ is an $(l,k_l)-$random code for $\fri$ with $\mu_l(\{(\mathcal{P}^l,\crr^l)  \}  )=1$ for some encoder-decoder pair $(\mathcal{P}^l,\crr^l)$ (This explains our requirement on $\sigma_l$ to contain all singleton sets) and $\mu_l(A)=0$ for any $A\in\sigma_l$ with $(\mathcal{P}^l,\crr^l)\notin A $. We will refer to such measures as point measures in what follows.
\begin{definition}
A non-negative number $R$ is a deterministically achievable entanglement transmission rate for the AVQC $\fri=\{\cn_s  \}_{s\in \bS}$ with error $\lambda\in[0,1]$ if it is achievable in the sense of Definition \ref{def:random-cap-ent-trans} for random codes  with \emph{point measures} $\mu_l$ and error $\lambda$.\\
The deterministic entanglement transmission capacity $\A_{\textup{det}}(\fri,\lambda)$ of $\fri$ with error $\lambda$ is given by
\begin{equation}
\A_{\textup{det}}(\fri,\lambda):=\sup\left\{R\in\rr_+:\begin{array}{l}R \textrm{ is an achievable entanglement trans-}\\ \textrm{mission rate for } \fri \textrm{ with deterministic codes and error }\lambda\end{array}\right\}.
 \end{equation}
\end{definition}
We now switch attention to message transmission. From the results in \cite{bn-correlation}, we know that average error criterion and maximal error criterion lead to the same capacity for AVQCs. Strictly speaking, this is a consequence of two facts: First, it does not really make sense to restrict the encoding functions to pure signal states in the quantum case, and second the two criteria are equivalent in the classical case as well, if one allows randomized encodings (see Ahlswede's paper \cite{ahlswede-elimination}, theorems 2 and 3).
\begin{definition}[Codes for message transmission over an AVQC]\label{def:message-trans-with-average-error}
Let $l\in\nn$. A random code for message transmission over $\fri$ is given by a probability measure $\gamma_l$ on the set $(CQ(M_l,\hr^{\otimes l})\times\mathcal M_{M_l},\Sigma_l)$, where $\Sigma_l$ again denotes a $\sigma-$algebra containing all singleton sets. Such $\sigma$-algebras exist, by arguments similar to those given in \cite{abbn} and\cite{bbjn}. A deterministic code is then given by a random code $\gamma_l$, where $\gamma_l$ is a point (also called Dirac) measure.
\end{definition}
\begin{definition}[Achievability]
A nonnegative number $R$ is called achievable with random codes with error $\lambda$ under the average error criterion if there exists a sequence $(\gamma_l)_{l\in\nn}$ of random codes satisfying both
\begin{align}
1)&\qquad\liminf_{l\to\infty}\min_{s^l\in\bS^l}\int\frac{1}{M_l}\sum_{i=1}^{M_l}\tr\{D_i\cn_{s^l}(\cP(i))\}d\gamma_l(\cP,\mathbf D)\geq1-\lambda\\
2)&\qquad\limsup_{l\to\infty}\frac{1}{l}\log M_l\geq R.
\end{align}
If the sequence $(\gamma_l)_{l\in\nn}$ can be chosen to consist of point measures only, then $R$ is called achievable with deterministic codes under the average error criterion.
\end{definition}
\begin{definition}[Message transmission capacities of an AVQC]
The corresponding capacities of $\fri$ are defined as
\begin{align}
\overline C_{\mathrm{det}}(\fri,\lambda)&:=\sup\left\{R:
\begin{array}{l}
R\ \mathrm{is\ achievable\ with\ deterministic\ codes}\\
\mathrm{under\ the\ average\ error\ criterion\ with\ error\ }\lambda
\end{array}\right\},\\
\overline C_{\mathrm{random}}(\fri,0)&:=\sup\left\{R:
\begin{array}{l}
R\ \mathrm{is\ achievable\ with\ random\ codes}\\
\mathrm{under\ the\ average\ error\ criterion\ with\ error\ }\lambda
\end{array}\right\}.
\end{align}
\end{definition}
The above definitions enable us to restate the main result of \cite{abbn} which connects the entanglement transmission capacities via the deterministic message transmission capacity $\overline C_{\mathrm{det}}(\cdot,0)$.
\begin{theorem}[Quantum Ahlswede dichotomy, cf. Theorem 1 in \cite{abbn}]\label{quant-ahlswede-dichotomy}
Let $\fri=\{\cn_s  \}_{s\in \bS}$ be a finite AVQC.
\begin{enumerate}
\item For $conv(\fri)$ given in (\ref{eq:conv-hull}) we have
  \begin{equation}\label{eq:ahlswede-dichotomy-1}
    \A_{\mathrm{random}}(\fri,0)=\lim_{l\to\infty}\frac{1}{l}\max_{\rho\in\cs(\hr^{\otimes l})}\inf_{\cn\in conv(\fri)}I_c(\rho, \cn^{\otimes l}).
  \end{equation}
\item Either $\overline C_{\mathrm{det}}(\fri,0)=0 $ or else $\A_{\mathrm{det}}(\fri,0)= \A_{\mathrm{random}}(\fri,0)$.
\end{enumerate}
\end{theorem}
It was also proven in \cite{abbn} that $\overline C_{\mathrm{det}}(\fri,0)=0$ holds if and only if the AVQC is \emph{symmetrizable} according to definition 39 in \cite{abbn}. This definition reads as follows:
\begin{definition}[Cf. definition 39 in \cite{abbn}]\label{def:c-symmetrizability}
Let $\bS$ be a finite set and $\fri=\{\cn_s  \}_{s\in\bS}$ an AVQC.
\begin{enumerate}
\item $\fri$ is called $l$-symmetrizable, $l\in\nn$, if for each finite set $\{\rho_1,\ldots,\rho_K  \}\subset \cs(\hr^{\otimes l})$, $K\in \nn$, there is a map $p:\{\rho_1,\ldots, \rho_K  \}\to \mathfrak{P}(\bS^l)$ such that for all $i,j\in\{1,\ldots, K  \}$
\begin{equation}\label{eq:c-symmetrizable}
\sum_{s^l\in\bS^l}p(\rho_i)(s^l)\cn_{s^l}(\rho_j)= \sum_{s^l\in\bS^l}p(\rho_j)(s^l)\cn_{s^l}(\rho_i)
\end{equation}
holds.
\item We call $\fri$ symmetrizable if it is $l$-symmetrizable for all $l\in\nn$.
\end{enumerate}
\end{definition}
As mentioned already, every AVQC $\fri$ is intimately connected to the compound quantum channel $\conv(\fri)$: the capacities of the AVQC $\fri$ are often given by the respective formulas for the corresponding compound quantum channels $\conv(\fri)$. This connection especially shows up in the proof and formulation of our theorem \ref{thm:random-code-reduction}, where we encounter the reliability functions of compound quantum channels. In order to define these, we first define codes, achievability and corresponding capacities for compound quantum channels:
\begin{definition}[Codes, achievability and capacities for compound quantum channels]\label{def:message-trans-for-compound-channel}
Let $l\in\nn$. A code $\mathfrak C_l$ for message transmission over the compound channel $\fri$ is given by a natural number $M_l$, an encoding $\cP:[M_l]\to\cs(\hr^{\otimes l})$ and a decoding $\mathbf D\in\mathcal M_{M_l}$. The error $\eps_l$ associated to the code is given by
\begin{align}
\eps_l:=1-\min_{s\in\bS}\frac{1}{M_l}\sum_{i=1}^{M_l}\tr\{D_i\cn_{s}^{\otimes l}\}.
\end{align}
A nonnegative number $R$ is called achievable for the compound channel $\fri$ with error $\lambda\in[0,1]$ if there exists a sequence $(\mathfrak C_l)_{l\in\nn}$ of codes for $\fri$ satisfying both $\limsup_{l\to\infty}\eps_l\leq\lambda$ and $\limsup_{l\to\infty}\frac{1}{l}\log M_l\geq R$.\\\\
The capacity of the compound channel $\fri$ with error $\lambda$ is defined as
\begin{align}
\overline C_{\mathrm{det}}^{\mathrm{comp}}(\fri,\lambda)&:=\sup\left\{R:
\begin{array}{l}
R\ \mathrm{is\ an\ achievable\ rate\ for\ the\ compound\ }\\
\mathrm{channel\ }\fri\ \mathrm{under\ the\ average\ error\ criterion}\textrm{\ and\ with\ error\ }\lambda
\end{array}\right\}.
\end{align}
By changing the reliability criterion from average to maximal error, we can define the corresponding capacity $C_{\mathrm{det}}^{\mathrm{comp}}$, and switching to entanglement- or strong subspace transmission defines the capacities $Q$ (see \cite{bbn-2}) and $Q_s$ in the obvious way.
\end{definition}
These definitions enable us now to define the corresponding reliability functions:
\begin{definition}[Reliability Functions] The reliability functions $\overline E_m,E_m,E_e,E_{s}:\mathcal C(\hr,\kr)\times\mathbb R_+\to\mathbb R_+$ are defined by
\begin{align}
\overline E_{m}(\fri,R)&:=\sup\left\{E:
\begin{array}{l}
\mathrm{There\ is\ }\eps>0\ \mathrm{and\ } N\in\nn\ \mathrm{such\ that}\ \mathrm{for\ all\ } l\geq N\ \mathrm{there\ is\ a\ code\ for}\\
\mathrm{message} \mathrm{\ transmission\ over\ the\ compound\ quantum\ channel\ }\fri\mathrm{\ satisfying}\\
\frac{1}{l}\log(M_l)\geq R-\eps\ \mathrm{and}\ \eps_l\leq2^{-l(E-\eps)},\ \mathrm{with}\ \eps_l\ \mathrm{being\ the\ average\ error.}
\end{array}\right\}\\
E_{m}(\fri,R)&:=\sup\left\{E:
\begin{array}{l}
\mathrm{There\ is\ }\eps>0\ \mathrm{and\ } N\in\nn\ \mathrm{such\ that}\ \mathrm{for\ all\ } l\geq N\ \mathrm{there\ is\ a\ code\ for}\\
\mathrm{message\ transmission\ over\ the\ compound\ quantum\ channel\ }\fri\mathrm{\ satisfying}\\
\frac{1}{l}\log(M_l)\geq R-\eps\ \mathrm{and}\ \eps_l\leq2^{-l(E-\eps)},\ \mathrm{with}\ \eps_l\ \mathrm{being\ the\ maximal\ error.}
\end{array}\right\}\\
E_{e}(\fri,R)&:=\sup\left\{E:
\begin{array}{l}
\mathrm{There\ is\ }\eps>0\ \mathrm{and\ } N\in\nn\ \mathrm{such\ that}\ \mathrm{for\ all\ } l\geq N\ \mathrm{there\ is\ a\ code\ for}\\
\mathrm{entanglement} \mathrm{\ transmission\ over\ the\ compound\ quantum\ channel\ }\fri\\
\mathrm{\ satisfying}\ \frac{1}{l}\log(k_l)\geq R-\eps\ \mathrm{and\ error}\ \eps_l\leq2^{-l(E-\eps)}.
\end{array}\right\}\\
E_{s}(\fri,R)&:=\sup\left\{E:
\begin{array}{l}
\mathrm{There\ is\ }\eps>0\ \mathrm{and\ } N\in\nn\ \mathrm{such\ that}\ \mathrm{for\ all\ } l\geq N\ \mathrm{there\ is\ a\ code\ for}\\
\mathrm{strong\ subspace\ transmission\ over\ the\ compound\ quantum\ channel\ }\fri\\
\mathrm{\ satisfying}\ \frac{1}{l}\log(k_l)\geq R-\eps\ \mathrm{and\ error}\ \eps_l\leq2^{-l(E-\eps)}.
\end{array}\right\}
\end{align}
\end{definition}
That above defined functions actually yield nonzero, finite values is not trivial in itself. It can, however, be explicitly read off the results in \cite{bbn-2} that $E_e(\fri,R)>0$ holds if $R< Q(\fri)$ (where $Q$ denotes the entanglement transmission capacity of the compound channel $\fri$ in the work \cite{bbn-2}), and that $E_m(\fri,R)$ can be larger than zero follows from \cite{bb-compound}.
\end{section}
\begin{section}{\label{sec:main-results}Main Results}
We now list our main results. If not specified otherwise, $\fri$ denotes a finite AVQC throughout the entire section.
\begin{theorem}\label{thm:CdetneqCrand}
Let $\fri$ consist of entanglement breaking channels that have the special form $\cn_s(\rho):=\sum_{x\in\bX}\tr\{\rho M_x\}\rho_{s,x}$, $s\in\bS$, for some finite set $\bS$ and POVM $\{M_i\}_{i=1}^M$ on $\hr$. The following two statements are true:
\begin{enumerate}
\item If there are probability distributions $\{p_x\}_{x\in\bX}\subset\mathfrak P(\bS)$ such that
\begin{align}
\sum_{s\in\bS}p_{x'}(s)\rho_{s,x}=\sum_{s\in\bS}p_x(s)\rho_{s,x'}\qquad\forall x,x'\in\bX,\label{eqn:example-is-symmetrizable}
\end{align}
then it holds $\overline C_{\mathrm{det}}(\fri,0)=0$.
\item There exists an example of an AVQC satisfying the above conditions which additionally has the property $\overline C_{\mathrm{random}}(\fri,0)>0$.
\end{enumerate}
\end{theorem}
\begin{remark}
Let us make a note on the intuition behind it. The channel $\fri$ is the concatenation of a stationary memoryless qc-channel (measurement) $\mathfrak W_1$ and an arbitrarily varying cq-channel $\mathfrak W_2$ given by the states $\{\rho_{s,x}\}_{s,x}$. This combination ensures that the channel itself is entanglement-breaking, whence its capacity has a one-shot formula and, even more important, it is $l$-symmetrizable for all $l\in\nn$ if and only if it is $1$-symmetrizable.\\
Using entangled inputs as signal states for $\fri$ results in mixtures of product states after the application of $\mathfrak W_1$, so $\mathfrak W_2$ sees a randomized code. But since we allow mixed inputs, this is equivalent to using just a randomized code with separable inputs for $\mathfrak W_2$. But on the subset of separable states signal states, $1-$symmetrizability is equivalent to $l$-symmetrizability for all $l\in\nn$, so no such code can transmit even a single bit with asymptotically vanishing error. Therefore, the deterministic capacity of $\mathfrak W$ has to be equal to zero.
\end{remark}
\begin{remark}
It is clear that the conjectured statement ``for all finite AVQCs, it holds that $\mathcal A_{\mathrm{det}}(\fri,0)=\mathcal A_{\mathrm{random}}(\fri,0)$'' is equivalent to saying that symmetrizability of a finite AVQC $\fri$ according to Definition 39 in \cite{abbn} implies that $\mathcal A_{\mathrm{random}}(\fri,0)=0$.\\
It is also clear that either one of the above would imply that $\mathcal A_{\mathrm{det}}$ is continuous, since $\mathcal A_{\mathrm{random}}$ is.
\end{remark}
\begin{theorem}\label{thm:Cdet-not-continuous}[Discontinuity of $\overline C_{\mathrm{det}}(\cdot,0)$]
The capacity function $\overline C_{\mathrm{det}}(\cdot,0):\mathcal C(\hr,\kr)\to\mathbb R_+$ is not continuous.\\
More precisely, let $\mathbb C^2=\linspan(\{e_1,e_2\})$ be naturally be embedded into $\mathbb C^3=\linspan(\{e_1,e_2,e_3\})$. Let the channel $\mathcal D_\eta\in\mathcal C(\mathbb C^2,\mathbb C^3)$ be defined through $\D_\eta(X):=(1-\eta)X+\eta\cdot\tr\{X\}\cdot\pi\ \forall X\in\mathcal B(\mathbb C^2)$. For any $\eta\in[0,1)$ the sequence $\fri_\lambda^\eta=\{\hat \cn_{s,\eta,\lambda}\}_{s\in\bS}$ of AVQCs defined by $\hat\cn_{s,\eta,\lambda}:=(1-\lambda)\mathcal D_\eta+\lambda\cn_s$ with $\{\cn_s\}_{s\in\bS}$ being the same set of channels as those constructed in Theorem \ref{thm:CdetneqCrand} satisfies
\begin{align}
\lim_{\lambda\to1}\overline C_{\mathrm{random}}(\fri_\lambda^\eta)=\overline C_{\mathrm{random}}(\fri,0)\geq0.5,\ \ \ \overline C_{\mathrm{det}}(\fri_\lambda^\eta,0)=\overline C_{\mathrm{random}}(\fri,0)\ \forall\lambda\in[0,1),\ \ \ \overline C_{\mathrm{det}}(\fri_1)=0.
\end{align}
On the other hand $\lim_{\lambda\to1}D_\lozenge(\fri_\lambda^\eta,\fri)=0$ for all $\eta\in[0,1]$, so that $\overline C_{\mathrm{det}}(\cdot,0)$ is discontinuous in the point $\fri_1$.
\end{theorem}
\begin{remark}
This is a first example of discontinuous behaviour of a quantum capacity other than the zero-error capacities. It is not clear to the authors yet, whether similar results could be proven for purely classical systems.\\
The example also highlights the stabilizing effect that is achieved by distribution of shared randomness in a communication system.
\end{remark}
\begin{theorem}[Positivity of $\overline C_{\mathrm{det}}$ is stable\label{thm:continuity-at-positivity-points}]
Let $\fri$ be a finite AVQC satisfying $\overline C_{\mathrm{det}}(\fri,0)>0$. There exists $\delta_0>0$ such that for all finite AVQCs $\fri'$ satisfying $D_\lozenge(\fri,\fri')\leq\delta_0$ it holds $\overline C_{\mathrm{det}}(\fri')>0$.
\end{theorem}
\begin{remark}
This result should be compared to the behaviour of the zero-error capacities, which are generically unstable (discontinuous at every point where they have a positive capacity). This comparison shows that communication systems exhibit a wide range of behaviour: Among them are continuous, discontinuous, stable and unstable ones.
\end{remark}
This theorem has the following two corollaries:
\begin{corollary}\label{cor:continuity-of-Cdet}
Let $\fri$ be a finite AVQC and $\overline C_{\mathrm{det}}(\fri,0)>0$. If $\lim_{l\to\infty}D_\lozenge(\fri,\fri_l)=0$ for some sequence $(\fri_l)_{l\in\nn}$ of finite AVQCs, then also $\lim_{l\to\infty}\overline C_{\mathrm{det}}(\fri_l)=\overline C_{\mathrm{det}}(\fri,0)$.\\
In other words, $\overline C_{\mathrm{det}}$ is lower-semi continuous, when restricted to (sequences of) finite AVQCs.
\end{corollary}
\begin{remark}
Of course, semi-continuity is a rather weak property (compared to continuity). In the history of quantum channel coding, lower semi-continuity of the entanglement transmission capacity of memoryless quantum channels has first been proven by \cite{keyl-werner-quantum-errors} in 2002. The first complete proof of continuity of the (entanglement transmission) capacity for a one-parameter family of quantum channels (erasure channels) was given in \cite{barnum-smolin-terhal} in 1998. Research on that line culminated in a proof by Leung and Smith \cite{leung-smith} in 2009, showing that the message- and entanglement transmission capacities of a memoryless quantum channel are continuous. Their result also holds in the presence of a wiretapper.\\
The results of Leung and Smith easily extend to the case of AVQCs: The capacity $\overline C_{\mathrm{random}}$, given by the formula in Lemma \ref{lemma:message-trans-AVQC}, is continuous. A method of proof can be picked up in \cite{abbn}, where a proof is given for the entanglement transmission capacity of an AVQC.
\end{remark}
\begin{corollary}\label{cor:continuity-of-Adet}
Let $\fri$ be a finite AVQC and $\overline C_{\mathrm{det}}(\fri,0)>0$. If $\lim_{l\to\infty}D_\lozenge(\fri,\fri_l)=0$ for some sequence $(\fri_l)_{l\in\nn}$ of finite AVQCs, then we have that
\begin{align}
\mathcal A_{\mathrm{det}}(\fri,0)=\mathcal A_{\mathrm{random}}(\fri,0)=\lim_{l\to\infty}\mathcal A_{\mathrm{det}}(\fri_l,0).
\end{align}
\end{corollary}
Since by Theorem \ref{thm:Cdet-not-continuous} we know that $\overline C_{\mathrm{det}}(\cdot,0)$ is not continuous, it makes sense to characterize the points of discontinuity:
\begin{theorem}[Characterization of points of discontinuity]\label{thm:characterizatino-of-discontinuity-points}
Let $\fri$ be an AVQC. The capacity function $\overline C_{\mathrm{det}}(\cdot,0)$ is discontinuous in the point $\fri$ if and only if $\overline C_{\mathrm{det}}(\fri,0)=0$, $\overline C_{\mathrm{random}}(\fri,0)>0$ and for every $\eps>0$ there exists a finite AVQC $\tilde\fri$ such that $D_\lozenge(\fri,\tilde\fri)<\eps$ and $\overline C_{\mathrm{det}}(\tilde\fri,0)>0$.
\end{theorem}
Our forthcoming Definition \ref{def:function-on-finite-AVQCs} in Subsection \ref{subsec:investigation-of-continuity} that is needed in the proof of Theorem  \ref{thm:continuity-at-positivity-points} motivated the formulation of the following observation:
\begin{lemma}\label{lemma:new-symmetrizability-condition}
Let $\fri=\{\cn_s\}_{s\in\bS}$ be a finite AVQC. If there exist $l\in\nn$ and functions $p,q:\cs(\hr^{\otimes l})\to\mathcal P(\bS^l)$ such that for all $\rho,\sigma\in\cs(\hr^{\otimes l})$ the equality
\begin{align}
\sum_{s^l\in\bS^l}p(\sigma)(s^l)\cn_{s^l}(\rho)=\sum_{s^l\in\bS^l}q(\rho)(s^l)\cn_{s^l}(\sigma)
\end{align}
holds, then $\conv(\{\cn_{s^l}\}_{s^l\in\bS^l})$ contains an entanglement breaking channel and, consequently, $\A_{\mathrm{rand}}(\fri,0)=0$.
\end{lemma}
\begin{theorem}[Random Code Reduction: finite error, finite randomness]\label{thm:random-code-reduction}
Let $\fri$ be a finite AVQC with $\overline C_{\mathrm{random}}(\fri,0)>0$ and $\lambda,\eps>0$, $0<R<\overline C_{\mathrm{random}}(\fri,0)$. There exist an $L(\fri,\lambda,R,\eps) \in\nn$ and $K(\fri,\lambda,R,\eps)\in\nn$ such that for all $l\geq L(\fri,\lambda,R,\eps)$ there are $K(\fri,\lambda,R,\eps)$ deterministic codes for the AVQC $\fri$ such that
\begin{align}\label{eq:random-code-reduction-a}
\frac{1}{K(\fri,\lambda,R,\eps)}\sum_{j=1}^{K(\fri,\lambda,R,\eps)} \frac{1}{M_l}\sum_{i=1}^{M_l}\tr(\cn_{s^l}(\rho_{i,j})D^l_{i,j})\geq1-\lambda \qquad  \forall s^l\in\mathbf S^l,\qquad\frac{1}{l}\log M_l\geq R-\eps.
\end{align}
It holds, with the abbreviation $E:=\overline E_{\mathrm{m}}(\conv(\fri),R)$,
\begin{align}
L(\fri,\lambda,R)=\min\{L:L(1-\frac{2\cdot|\bS|}{E-\eps}\cdot\frac{1}{L}\log(L)\geq\frac{2}{E-\eps}\log(\frac{1}{\lambda}\cdot\frac{4}{E-\eps})\},
\end{align}
and the quantitiy $K(\fri,\lambda,R)$ is given by
\begin{align}
K(\fri,\lambda,R,\eps)=\frac{1}{\lambda}\cdot\frac{8\cdot\log|\bS|}{E-\eps}.
\end{align}
The same statement holds with average error criterion replaced by entanglement fidelity and $\overline E_{\mathrm{m}}$ by $E_{\mathrm{e}}$.
\end{theorem}
\begin{remark}
It is clear that above statement is especially interesting for the message transmission capacity of an AVQC, and there only in the case when the deterministic capacity vanishes but the randomness assisted one does not.\\
As a very rough approximation, one may use the scaling law $L(\fri,\lambda,R)\approx\frac{2}{\overline E_m(\conv(\fri),R)-\eps}\log(\frac{1}{\lambda}\cdot\frac{4}{\overline E_m(\conv(\fri),R)-\eps})$. It is clear that both $L$ and $K$ from above theorem are sub-optimal, even with the techniques used in this paper. However, their scaling with $\lambda$ does not depend on the choice of constants in our proof. For fixed $\fri$ and rate $R$, this means that the block-length needed to achieve a certain error $\lambda$ roughly scales as $\log(1/\lambda)$, and the randomness as $1/\lambda$.
\end{remark}
\begin{theorem}\label{thm:AsequalsA}
Let $\fri$ be a finite AVQC and $\lambda\in[0,1]$. Then both
\begin{align}
\mathcal A_{\mathrm{det}}(\fri,\lambda)=\mathcal A_{\mathrm{s,det}}(\fri,\lambda)\qquad\mathrm{and}\qquad\mathcal A_{\mathrm{random}}(\fri,\lambda)=\mathcal A_{\mathrm{s,random}}(\fri,\lambda).
\end{align}
\end{theorem}
\begin{remark}
It should, of course, be noted that we expect this picture to change once finite block-lengths are brought into the game. We leave this for future work.
\end{remark}
\end{section}
\begin{section}{\label{sec:proofs}Proofs}
We now give the proofs of our results, in the same order they appeared in Section \ref{sec:main-results}.
\begin{subsection}{\label{subsec:randomness-advantageous}Using randomness is advantageous for message transmission over AVQCs}
This subsection is devoted to the proof of Theorem \ref{thm:CdetneqCrand}.
\begin{proof}
We will first construct and AVQC having the special structure introduced in Theorem \ref{thm:CdetneqCrand}, then prove that every $AVQC$ having that special form is symmetrizable, and finally estimate the random capacity for the one special choice we made:\\
1) A particular example is given by $\bX=\bS=\{0,1\}$ and with $\kr=\mathbb C^3$. Take an arbitrary POVM $\{M_x\}_{x\in\bX}$ and three different states $\sigma_1,\sigma_2,\sigma_3$. This will be enough to prove that $\overline C_{\mathrm{det}}(\fri,0)=0$. In order to get an easy estimate on the random capacity, we have to make a more explicit choice for the $\sigma_i$. The choice $\sigma_1=|e_1\rangle\langle e_1|$, $\sigma_2=|e_2\rangle\langle e_2|$, $\sigma_3=|e_3\rangle\langle e_3|$, $M_i=|e_i\rangle\langle e_i|$, $i=1,2$ will do. Set
\begin{align}
\rho_{1,1}=\sigma_1,\ \ \ \rho_{1,2}=\sigma_3,\ \ \ \rho_{2,1}=\sigma_3,\ \ \ \rho_{2,2}=\sigma_2.
\end{align}
Then obviously the choice $p_1(i)=\delta(1,i)$, $p_2(i)=\delta(2,i)$ (with $\delta(\cdot,\cdot)$ being the usual Kronecker-delta function) fulfills equation (\ref{eqn:example-is-symmetrizable}).\\
2) $\fri$ is $1$-symmetrizable: Given the states $\{\nu_x\}_{x\in\bX}$, we choose the distributions $\{q_x\}_{x\in\bX}$ defined by
\begin{align}
q_x(s):=\sum_{\tilde x\in\bX}p_{\tilde x}(s)\tr\{M_{\tilde{x}}\nu_x\},
\end{align}
then symmetrizability holds by the following calculation:
\begin{align}
\sum_{s\in\bS}q_{x'}(s)W_s(\nu_x)&=\sum_{s\in\bS}q_{x'}(s)\sum_{\tilde x\in\bX}\tr\{\nu_xM_{\tilde{x}}\}\rho_{s,\tilde x}\\
&=\sum_{\tilde x\in\bX}\tr\{\nu_xM_{\tilde{x}}\}\sum_{s\in\bS}q_{x'}(s)\rho_{s,\tilde x}\\
&=\sum_{\tilde x\in\bX}\tr\{\nu_xM_{\tilde{x}}\}\sum_{s\in\bS}\sum_{\hat x\in\bX}p_{\hat x}(s)\tr\{M_{\hat{x}}\nu_{x'}\}\rho_{s,\tilde x}\\
&=\sum_{\tilde x\in\bX}\tr\{\nu_xM_{\tilde{x}}\}\sum_{\hat x\in\bX}\tr\{M_{\hat{x}}\nu_{x'}\}\sum_{s\in\bS}p_{\hat x}(s)\rho_{s,\tilde x}\\
&=\sum_{\tilde x\in\bX}\tr\{\nu_xM_{\tilde{x}}\}\sum_{\hat x\in\bX}\tr\{M_{\hat{x}}\nu_{x'}\}\sum_{s\in\bS}p_{\tilde x}(s)\rho_{s,\hat x}\\
&=\sum_{s\in\bS}\left(\sum_{\tilde x\in\bX}\tr\{\nu_xM_{\tilde{x}}\}p_{\tilde x}(s)\right)\sum_{\hat x\in\bX}\tr\{M_{\hat{x}}\nu_{x'}\}\rho_{s,\hat x}\\
&=\sum_{s\in\bS}q_x(s)\sum_{\hat x\in\bX}\tr\{M_{\hat{x}}\nu_{x'}\}\rho_{s,\hat x}\\
&=\sum_{s\in\bS}q_x(s)W_s(\nu_{x'}).
\end{align}
3) We now want to get a nonzero lower bound on the random capacity of $\overline C_{\mathrm{random}}(\fri,0)$. To do so, we employ the results \cite{ahlswede-blinovsky} of Ahlswede and Blinovsky:\\
Fix the allowed input states $\cP(i)$ in the encoding of messages into $\fri$ to arbitrary tensor products of a set $\{\rho_1,\ldots,\rho_{\dim(\hr)}\}$ of pure states. It is clear that such a strategy is hopelessly sub-optimal. But by restriction to this model, we get from Theorem 1 in \cite{ahlswede-blinovsky} the lower bound
\begin{align}\label{eqn:one-shot-capacity-1}
\overline C_{\mathrm{random}}(\fri,0)\geq\max_{p\in\mathfrak P([\dim(\hr)])}\min_{\cn\in\conv(\fri)}S(\sum_{k=1}^Kp(k)\cn(\rho_k))-\sum_{k=1}^Kp(k)S(\cn(\rho_k)).
\end{align}
But the function $(p,\cn)\to S(\sum_{k=1}^Kp(k)\cn(\rho_k))-\sum_{k=1}^Kp(k)S(\cn(\rho_k))$ is concave in $p$ and convex in $\cn$ and both optimization procedures in (\ref{eqn:one-shot-capacity-1}) are over convex compact sets, so this translates to
\begin{align}\label{eqn:one-shot-capacity-2}
\overline C_{\mathrm{random}}(\fri,0)\geq\min_{\cn\in\conv(\fri)}\max_{p\in\mathfrak P([K])}S(\sum_{k=1}^Kp(k)\cn(\rho_k))-\sum_{k=1}^Kp(k)S(\cn(\rho_k)),
\end{align}
by the minimax-Theorem. Therefore,
\begin{align}
\overline C_{\mathrm{random}}(\mathfrak W,0)&\geq\min_{\cn\in\conv(\fri)}\max_{\{q_i,\sigma_i\}}\chi(\{q_i,\sigma_i\}_i,\cn)\\
&\geq\min_{\cn\in\conv(\{\cn_1,\cn_2\})}\chi(\{\frac{1}{2},|e_i\rangle\langle e_i\}_{i=1}^2,\cn)\\
&=\min_{t\in[0,1]}H([t/2,(1-t)/2,1/2,0])-t\cdot h(1/2)-(1-t)\cdot h(1/2)\\
&=1-h(t)/2\\
&\geq1/2\\
&>0.
\end{align}
\end{proof}
We use the opportunity to state the following additional Lemma:
\begin{lemma}[Message Transmission Capacity of an AVQC]\label{lemma:message-trans-AVQC}
Let $\fri$ be a finite AVQC. Then
\begin{align}
\overline C_{\mathrm{random}}(\fri,0)=\lim_{l\to\infty}\frac{1}{l}\min_{\cn\in\conv(\fri)}\chi(\cn^{\otimes l}).
\end{align}
\end{lemma}
\begin{proof}
The converse statement ``LHS$\leq$RHS'' follows from observing first that each code for the AVQC $\fri$ is also a code for the compound quantum channel $\fri$, with the same average error.\\
The inequality ``LHS$\geq$RHS'' follows from using any of the existing proofs for the converse theorem for the memoryless stationary quantum channel in combination with the minimax theorem.\\
The direct part is a consequence of \cite{ahlswede-blinovsky}, together with the usual blocking strategies. For more information on how these steps are carried out in detail, we refer to \cite{abbn}, where this is done for entanglement transmission.
\end{proof}
\end{subsection}
\begin{subsection}{\label{subsec:investigation-of-continuity}Investigation of Continuity}
In this subsection, we prove Theorem \ref{thm:Cdet-not-continuous}, Theorem \ref{thm:characterizatino-of-discontinuity-points}, Theorem \ref{thm:continuity-at-positivity-points}, its corollaries and, finally, Lemma \ref{lemma:new-symmetrizability-condition}.\\
Let us get started. It is important to note here, that a crucial ingredient to the proof of Theorem \ref{thm:CdetneqCrand} is that $D_\eta$ is \emph{not} entanglement breaking for $\eta\neq1$, while $\fri$ from Theorem \ref{thm:CdetneqCrand} is. At the same time, the $D_\eta$ have a very simple structure. This is what enables us to prove non-symmetrizability for all $\lambda,\eta\in[0,1)$:
\begin{proof}[Proof of Theorem \ref{thm:Cdet-not-continuous}]
By Theorem \ref{thm:CdetneqCrand} we know that $\fri_1^\eta=\fri$ is symmetrizable for all $\eta\in[0,1]$. It is, as an additional fact, clear that $\fri_\lambda^1$ is symmetrizable for all $\lambda\in[0,1]$.\\
We are going to show that $\fri_\lambda^\eta$ is non-symmetrizable for all $\lambda,\eta\in[0,1)$:\\
Assume that, to a given pair $\rho,\sigma\in\cs(\mathbb C^2)$ of input states there are $p,q\in\mathcal P(\bS)$ such that
\begin{align}
\sum_{s\in\bS}p(s)\hat\cn_{s,\eta,\lambda}(\rho)=\sum_{s\in\bS}q(s)\hat\cn_{s,\eta,\lambda}(\sigma).
\end{align}
Rearranging terms, we get that this is equivalent to
\begin{align}
\lambda\sum_{s\in\bS}[q(s)\cn_s(\sigma)-p(s)\cn_s(\rho)]-(1-\lambda)\eta\pi=(1-\lambda)(1-\eta)(\rho-\sigma).
\end{align}
Especially, the latter equality would have to hold for the entries of the respective matrices in the standard orthonormal basis $\{e_i\}_{i=1}^3$ we may choose, so that a look at the off-diagonal entries reveals that (since the channels $\cn_s$ are all entanglement-breaking), the equality
\begin{align}
0=(1-\lambda)(1-\eta)\langle e_1,(\rho-\sigma)e_2\rangle
\end{align}
would have to hold. Since both $\lambda\neq1$ and $\eta\neq1$ we see that this would require $\langle e_1,(\rho-\sigma)e_2\rangle=0$, and that is clearly not valid for all choices $\rho,\sigma$ of input states. Thus $\fri_\lambda^\eta$ is non-symmetrizable for $\lambda,\eta\in[0,1)$.
\end{proof}
We now turn to the proof of Theorem \ref{thm:continuity-at-positivity-points} and its corollaries, after which we give a proof of Theorem \ref{thm:characterizatino-of-discontinuity-points}. Then, we prove Lemma \ref{lemma:new-symmetrizability-condition}. A preliminary definition is needed:
\begin{definition}\label{def:function-on-finite-AVQCs}
To any given $l\in\nn$, define a nonnegative function $F_l$ on the set of all finite subsets of $\mathcal C(\hr^{\otimes l},\kr^{\otimes l})$ through
\begin{align}
\{\cn_i\}_{i\in\mathcal I}\mapsto\max_{\rho,\sigma\in\cs(\hr^{\otimes l})}\min_{q,p\in\mathcal P(\mathcal I)}\|\sum_{i\in\mathcal I}[p(i)\cn_i(\rho)-q(i)\cn_i(\sigma)\|_1.
\end{align}
\end{definition}
\begin{remark}
This function is well-defined, by compactness of $\cs(\hr^{\otimes l})$ and $\mathcal P(\mathcal I)$.
\end{remark}
\begin{remark}
For any finite AVQC $\fri$ the statements "$\fri$ is $l$-symmetrizable" and "$F_l(\fri)=0$" are equivalent.
\end{remark}
\begin{proof}[Proof of Theorem \ref{thm:continuity-at-positivity-points}]
Let $\fri=\{\cn_s\}_{s\in\bS}$. Let $q,p\in\mathcal P(\bS)$ be arbitrary and $\rho,\sigma\in\cs(\hr)$. Let $\fri'$ be a finite AVQC satisfying
\begin{align}\label{eqn:hausdorff-abstand-endlicher-avqcs}
D_\lozenge(\fri,\fri')\leq\eta
\end{align}
for some $\eta>0$ that will be specified later. We write $\fri'=\{\cn'_x\}_{x\in\bX}$, and by using the same channel twice if necessary we may assume that $\bX=\bS$ holds. The inequality (\ref{eqn:hausdorff-abstand-endlicher-avqcs}) implies that there are functions $f,g:\bS\to\bS$ such that for every $s\in\bS$ we have
\begin{align}
\|\cn'_s-\cn_{f(s)}\|_\lozenge\leq\eta,\qquad \|\cn'_{g(s)}-\cn_{s}\|_\lozenge\leq\eta,\qquad\forall s\in\bS.
\end{align}
It then follows from convexity of $\|\cdot\|_\lozenge$ that
\begin{align}
\|\sum_{s\in\bS}p(s)[\cn_{s}-\cn'_{g(s)}]\|_\lozenge\leq\eta,\\
\|\sum_{s\in\bS}q(s)[\cn_{s}-\cn'_{g(s)}]\|_\lozenge\leq\eta.
\end{align}
Define $\tilde q,\tilde p\in\mathcal P(\bS)$ by
\begin{align}
\tilde q(s):=\sum_{s':g(s')=s}q(s),\qquad\tilde p(s):=\sum_{s':f(s')=s}p(s).
\end{align}
It then holds that
\begin{align}
\sum_{s\in\bS}q(s)\cn_{g(s)}=\sum_{s\in\bS}\tilde q(s)\cn_s,\qquad \sum_{s\in\bS}p(s)\cn_{f(s)}=\sum_{s\in\bS}\tilde p(s)\cn_s.
\end{align}
Thus
\begin{align}
\|\sum_{s\in\bS}[\tilde q(s)\cn_s(\sigma)-\tilde p(s)\cn_s(\rho)]\|_1&=\|\sum_{s\in\bS}[\tilde q(s)\cn_s(\sigma)-q(s)\cn'_s(\sigma)+q(s)\cn'_s(\sigma)+\\
&\qquad+p(s)\cn'_s(\rho)-p(s)\cn'_s(\rho)-\tilde p(s)\cn_s(\rho)]\|_1\\
&\leq\|\sum_{s\in\bS}[\tilde q(s)\cn_s(\sigma)-q(s)\cn_s'(\sigma)]\|_1+\|\sum_{s\in\bS}[p(s)\cn_s'(\rho)-\\
&\qquad-\tilde p(s)\cn_s(\rho)]\|_1+\|\sum_{s\in\bS}[q(s)\cn_s'(\sigma)-p(s)\cn_s'(\rho)]\|_1\\
&\leq2\eta+\|\sum_{s\in\bS}[q(s)\cn_s'(\sigma)-p(s)\cn_s'(\rho)]\|_1.
\end{align}
This holds for all $\rho,\sigma,p,q$, and minimizing the left hand side we especially get that
\begin{align}
\min_{q',p'\in\mathcal P(\bS))}\|\sum_{s\in\bS}[q'(s)\cn_s(\sigma)-p'(s)\cn_s(\rho)]\|_1\leq2\eta+\textcolor{red}{\blacksquare}\|\sum_{s\in\bS}[q(s)\cn_s'(\sigma)-p(s)\cn_s'(\rho)]\|_1.
\end{align}
This implies
\begin{align}
\min_{q,p\in\mathcal P(\bS)}\|\sum_{s\in\bS}[q(s)\cn_s(\sigma)-p(s)\cn_s(\rho)]\|_1\leq2\eta+\min_{q,p\in\mathcal P(\bS)}\|\sum_{s\in\bS}[q(s)\cn_s'(\sigma)-p(s)\cn_s'(\rho)]\|_1,
\end{align}
and from there it clearly follows that $F_1(\fri)\leq2\eta +F_1(\fri')$. Now choose $\eta=F_1(\fri)/4$, then we get \textcolor{red}{$\blacksquare$} $F_1(\fri)\leq F_1(\fri)/2+F_1(\fri')$. But that directly implies
\begin{align}
F_1(\fri)/2\leq F_1(\fri'),
\end{align}
and this in turn implies $\overline C_{\mathrm{det}}(\fri',0)>0$.
\end{proof}
\begin{proof}[Proof of Corollary \ref{cor:continuity-of-Cdet}]
It is clear that there is an $L\in\nn$ such that for all $l\geq L$ we have $\overline C_{\mathrm{det}}(\fri_l,0)>0$. But if $\overline C_{\mathrm{det}}(\fri_l,0)>0$, then $\overline C_{\mathrm{det}}(\fri_l,0)=\overline C_{\mathrm{random}}(\fri_l,0)$ and the latter quantity is continuous by the formula in Lemma \ref{lemma:message-trans-AVQC} and the results \cite{leung-smith} of Leung and Smith (for their application to quantum capacity formulae of AVQCs, \cite{abbn} is a good reference), so (reading equalities from left to right)
\begin{align}
\lim_{l\to\infty}\overline C_{\mathrm{det}}(\fri_l,0)=\lim_{l\to\infty}\overline C_{\mathrm{random}}(\fri_l,0)=\overline C_{{random}}(\fri,0).
\end{align}
\end{proof}
\begin{proof}[Proof of Corollary \ref{cor:continuity-of-Adet}]
The l.h.s. equality has been proven in (\cite{abbn}, Theorem 1), and the r.h.s. inequality is proven in the same way as Corollary \ref{cor:continuity-of-Cdet}: For all $l\geq L$, we have that $\mathcal A_{\mathrm{det}}(\fri_l)=\mathcal A_{\mathrm{random}}(\fri_l)$, and the latter quantity is continuous with respect to $D_\lozenge$ (this was implicitly proven in \cite{abbn}, see equations (98) and (99)), proving the corollary.
\end{proof}
In the following lines, we prove Theorem \ref{thm:characterizatino-of-discontinuity-points}.
\begin{proof}[Proof of Theorem \ref{thm:characterizatino-of-discontinuity-points}]
Let $\fri$ ($|\fri|<\infty$!) be a discontinuity point of $\overline C_{\mathrm{det}}(\cdot,0)$. By Corollary \ref{cor:continuity-of-Cdet} we know that $\overline C_{\mathrm{det}}(\fri,0)=0$ via contradiction. This implies that there exists a $\delta>0$ such that $\overline C_{\mathrm{det}}(\fri',0)\geq\delta$ for all $\fri'$ satisfying $0<D_\lozenge(\fri',\fri)\leq\delta$.\\
We still have to show that $\overline C_{\mathrm{random}}(\fri,0)>0$ follows, so assume the contrary: $\overline C_{\mathrm{random}}(\fri,0)=0$. Then, there would certainly exist a finite AVQC $\fri'$ satisfying $D_\lozenge(\fri',\fri)\leq\delta$ and $\overline C_{\mathrm{random}}(\fri',0)<\delta$. But then
\begin{align}
\delta\leq \overline C_{\mathrm{det}}(\fri',0)\leq \overline C_{\mathrm{random}}(\fri',0)<\delta
\end{align}
would have to hold, a clear contradiction. Thus, $\overline C_{\mathrm{random}}(\fri,0)>0$.\\
On the other hand, let $\overline C_{\mathrm{det}}(\fri,0)=0$, $\overline C_{\mathrm{random}}(\fri,0)>0$ and assume that for all $\delta>0$ there exists a finite $\fri_\delta$ satisfying $0<D_\lozenge(\fri_\delta,\fri)\leq\delta$ such that $\overline C_{\mathrm{det}}(\fri_\delta,0)>0$. Then by continuity of $\overline C_{\mathrm{random}}(\cdot,0)$ we have for some $\eps>0$ that $\overline C_{\mathrm{random}}(\fri_\delta,0)>\eps$ for all those $\fri_\delta$ for which $\delta$ is small enough ($0<\delta<\delta_0$, say).\\
In addition, it is clear that for $0<\delta<\delta_0$ we also have $\overline C_{\mathrm{det}}(\fri_\delta,0)=\overline C_{\mathrm{random}}(\fri_\delta,0)$, whence $\overline C_{\mathrm{det}}(\fri_\delta,0)>\eps$. This proves that $\overline C_{\mathrm{det}}(\cdot,0)$ is discontinuous in the point $\fri$.
\end{proof}
\begin{proof}[Proof of Lemma \ref{lemma:new-symmetrizability-condition}]
Clearly, by assumption we have that $\fri$ is $l$-symmetrizable. Fix $\sigma\in\cs(\hr^{\otimes l})$ and define $\cn_\sigma:=\sum_{s^l\in\bS^l}p[\sigma](s^l)\cn_{s^l}$. Then it holds, for all $\rho\in\cs(\hr^{\otimes l})$,
\begin{align}
\cn_\sigma(\rho)=\sum_{s^l\in\bS^l}p[\sigma](s^l)\cn_{s^l}(\rho)=\sum_{s^l\in\bS^l}q[\rho](s^l)\cn_{s^l}(\sigma).
\end{align}
Especially for every two states $\rho,\rho'$ and $\lambda\in(0,1)$ we get
\begin{align}\label{eqn:important-equality-in-new-symmetrizability-condition}
\sum_{s^l\in\bS^l}q[\lambda\rho+(1-\lambda)\rho'](s^l)\cn_{s^l}(\sigma)&=\sum_{s^l\in\bS^l}(\lambda q[\rho](s^l)+(1-\lambda)q[\rho'](s^l))\cn_{s^l}(\sigma).
\end{align}
The operators $\cn_{s^l}(\sigma)$ may be linearly dependent, so that no general conclusion can be drawn yet, we have to go one step further: The set $\conv(\cn_{s^l}(\sigma))$ has extremal points $\{\cn_t\}_{t\in T}$ for some subset $T\subset\bS^l$ and we may choose probability distributions $\{r(\cdot|s^l)\}_{s^l\in \bS^l}\subset\mathfrak P(T)$ such that we get
\begin{align}
\forall\ s^l\in\bS^l:\qquad \cn_{s^l}(\sigma)=\sum_{t\in T}r(t|s^l)\cn_{s^l}(\sigma).
\end{align}
Defining $\tilde q:\cs(\hr^{\otimes l})\to\mathfrak P(T)$ pointwise through $\tilde q[\rho](t):=\sum_{s^l\in\bS^l}q[\rho](s^l)r(t|s^l)$ this transforms the statement (\ref{eqn:important-equality-in-new-symmetrizability-condition}) into
\begin{align}
\sum_{t\in T}\tilde q[\lambda\rho+(1-\lambda)\rho'](t)\cn_t(\sigma)&=\sum_{t\in T}(\lambda \tilde q[\rho](t)+(1-\lambda)\tilde q[\rho'](t))\cn_t(\sigma).
\end{align}
Since convex combinations of extremal points are unique, this proves that for all $\rho,\rho'\in\cs(\hr^{\otimes l})$ and $\lambda\in(0,1)$ we have
\begin{align}
\tilde q[\lambda\rho+(1-\lambda)\rho']&=\lambda\tilde q[\rho]+(1-\lambda)\tilde q[\rho'].
\end{align}
In a next step we can now linearly extend the functional $\tilde q$ to all of $\mathcal B(\hr^{\otimes l})$: the (complex) linear span of $\cs(\hr^{\otimes l})$ obviously contains the set of all nonnegative operators, whose linear span clearly contains the self-adjoint operators, and their linear span finally is $\mathcal B(\hr^{\otimes l})$.\\
We have thus (since $T\subset\bS^l$) constructed a linear functional $\tilde q$ such that $\tilde q(\cs(\hr^{\otimes l}))\subset\mathfrak P(\bS^l)$ and
\begin{align}
\forall\ \rho\in\cs(\hr^{\otimes l}):\qquad\cn_\sigma(\rho)=\sum_{s^l\in\bS^l}\tilde q[\rho](s^l)\cn_{s^l}(\sigma).
\end{align}
Since each $\tilde q[\cdot](t)$ is a linear functional it may be represented as $\tilde q[\cdot](t)=\tr\{A_t\cdot\}$ for some operator $A_t$, and since it is both positivity preserving and trace preserving the operators $A_t$ can be chosen to be nonnegative and must satisfy the condition $\sum_{t\in T}A_t=\eins_{\hr^{\otimes l}}$, whence the form a POVM and it holds
\begin{align}
\forall\ \rho\in\cs(\hr^{\otimes l}):\qquad\cn_\sigma(\rho)=\sum_{s^l\in\bS^l}\tr\{A_{s^l}\rho\}\cn_{s^l}(\sigma).
\end{align}
This is equivalent to the statement that $\cn_\sigma\in\conv(\{\cn_{s^l}\})$ is entanglement breaking by the results of \cite{horodecki-shor-ruskai}, which implies that $\A_{\mathrm{random}}(\{\cn_{s^l}\}_{s^l\in\bS^l})=0$ by the results of \cite{abbn}, which is equivalent to $\A_{\mathrm{random}}(\fri,0)=0$ by the very definition of the channel model.
\end{proof}
\end{subsection}
\begin{subsection}{\label{sec:finite-randomness}Finite amounts of randomness, finite errors and blocklengths}
In this subsection, we give the proofs of the results concerning finite resources. Namely, we are interested here in the scaling laws that connect finite errors $\lambda$ and the amount of randomness needed to achieve such an error. At the same time, we are able to give estimates on the block-length that is needed to achieve a certain error.\\
In addition, we present the following Lemma which is an immediate drop-off from the results in \cite{abbn}:
\begin{lemma}[A partial strong converse for message transmission over finite AVQCs]
For every finite AVQC $\fri$ which is symmetrizable and for every $\lambda\in[0,1/2)$ we have $\overline C(\fri,\lambda)=\overline C(\fri,0)$.
\end{lemma}
\begin{proof}
This follows directly from the proof of Theorem 40 in \cite{abbn}, with a slightly better estimate in between equations (202) and (203).
\end{proof}
We now turn to the original goal of this subsection, the proof of Theorem \ref{thm:random-code-reduction}:
\begin{proof}[Proof of Theorem \ref{thm:random-code-reduction}]
Let $l\in\nn$, $\eps>0$ and $\delta:=\overline E_m(\conv(\fri),\overline C_{\mathrm{compound}}(\conv(\fri)))-\eps$. If only $l$ is large enough, then by definition of $E$ there is a code $\{\sigma_i,P_i\}_{i=1}^{M_l}$ with $M_l=\lfloor2^{l(\overline C_{\mathrm{compound}}(\conv(\fri))-\eps)}\rfloor$ such that
\begin{align}
\frac{1}{M_l}\sum_{i=1}^{M_l}\tr\{P_i\cn^{\otimes n}(\sigma_i)\}\geq1-2^{-l\delta}\qquad\forall\cn\in\conv(\fri)
\end{align}
and that $\delta$ is the largest possible such value over all choices of codes at the same rate and same blocklength. By application of the robustification technique (find the original in \cite{ahlswede-coloring}, Theorem 6, an improved version in \cite{ahlswede-gelfand-pinsker} or read Theorem 28 in \cite{abbn}) we get a random code $\mu_l$ at the same rate such that
\begin{align}
e(\mu_l,\cn):=\min_{s^l\in\bS^l}\int \frac{1}{M_l}\sum_{i=1}^{M_l}\tr(\cn_{s^l}(\rho_i)D_i^l)d\mu_l((\rho_i,D_i^l)_{i=1}^{M_l})\ge
1-(l+1)^{|\bS|}2^{-l\delta}.
\end{align}
We will now use the abbreviation $\eps_l:=(l+1)^{|\bS|}2^{-l\delta}$. For a fixed $K\in\nn$, consider $K$ independent random variables $\Lambda_i$ with values in $((\cs(\hr^{\otimes l})^{M_l})\times\M_{M_l}(\hr^{\otimes l}))$
which are distributed according to
$\mu_l$.\\
Define, for each $s^l\in\bS^l$, the function $p_{s^l}:((\cs(\hr^{\otimes l})^{M_l})\times\M_{M_l}(\hr^{\otimes l}))\rightarrow[0,1]$,\\
$(\rho_1,\ldots,\rho_{M_l},D_1^l,\ldots,D_{M_l}^l)\mapsto\frac{1}{M_l}\sum_{i=1}^{M_l}\tr(\cn_{s^l}(\rho_i)D_i^l)$.\\
We get, by application of Markovs inequality, for every $s^l\in\bS^l$, and every $r\geq0$:
\begin{align}
\mathbb P(1-\frac{1}{K}\sum_{j=1}^{K}p_{s^l}(\Lambda_j)\geq\lambda)&= \mathbb P(2^{rK- r\sum_{j=1}^Kp_{s^l}(\Lambda_j)}\geq2^{rK\lambda})\\
&\leq2^{-rK\lambda}\mathbb E(2^{ (rK-r\sum_{j=1}^Kp_{s^l}(\Lambda_j))}).
\end{align}
The $\Lambda_i$ are independent and it holds $2^{rt}\leq1+t2^r$ for every $t\in[0,1]$ and $r\geq0$ as well as $\log(1+x)\leq2x$ for $x\geq0$ and so we get
\begin{align}
 \mathbb P(1-\frac{1}{K}\sum_{j=1}^{K}p_{s^l}(\Lambda_j)\geq\lambda)&\leq 2^{-rK\lambda}\mathbb E(2^{rK-r\sum_{j=1}^Kp_{s^l}(\Lambda_j)})\\
&=2^{-rK\lambda}\mathbb E(2^{r\sum_{j=1}^K(1-p_{s^l}(\Lambda_j))})\\
&=2^{-rK\lambda}\mathbb E(2^{(r-rp_{s^l}(\Lambda_1))})^K\\
&\leq2^{-rK\lambda}\mathbb E(1+(1-p_{s^l}(\Lambda_1)2^r))^K\\
&\leq2^{-rK\lambda}\mathbb E(1+\eps_l2^r)^K\\
&\leq2^{-rK\lambda}2^{K\eps_l2^r}\\
&=2^{-K(r\lambda-\eps_l2^r)}.
\end{align}
Therefore, and with the choice $r=l\delta/2$, we get by application of the union bound that
\begin{align}
  \mathbb P(\frac{1}{K}\sum_{j=1}^{K}p_{s^l}(\Lambda_j)\geq1-\lambda\ \forall s^l\in\bS^l)&\geq1-|\bS|^l2^{-K(l\delta\lambda/2-(l+1)^{|\bS|}2^{-l\delta/2})}\\
  &=1-2^{-l(K(\delta\lambda/2-\frac{(l+1)^{|\bS|}}{l}2^{-l\delta/2})-\log(|\bS|))}.
\end{align}
But this shows that, for $K>\log(|\bS|)/(\frac{\delta\lambda}{2}-\frac{(l+1)^{|\bS|}}{l}2^{-l\delta/2})$, the above probability is larger than zero, so there exists a realization
$\Lambda_1,\ldots,\Lambda_{l^n}$ such that
\begin{equation}
\min_{s^l\in\bS^l}\frac{1}{K}\sum_{i=1}^{K}\frac{1}{M_l}\tr(\cn_{s^l}(\rho_i)D_i^l)\geq1-\lambda.
\end{equation}
The important conclusion we draw from this result is that, for $l$ large enough ($l\geq L(\delta,\lambda,|\bS|)$ for some $L(\delta,\lambda,|\bS|)\in\nn$), $\frac{(l+1)^{|\bS|}}{l}2^{-l\delta/2}<\delta\lambda/4$, and then for $K=8\log(|\bS|)/\delta\lambda$ we are the guaranteed existence of random codes that only use a finite amount of randomness. This proves the theorem.
\end{proof}
\begin{proof}[Proof of Theorem \ref{thm:AsequalsA}]
At last, let us consider the differences between entanglement- and strong subspace transmission. Let $l\in\nn$ be arbitrary. It has been shown in \cite{abbn}, their Lemma 21, that every subspace $\fr_l$ for entanglement transmission (with shared randomness or without) over an (even non-finite) AVQC $\fri$ with error $\lambda_l\in[0,1]$ contains a subspace $\hat\fr_l\subset\fr_l$ for strong subspace transmission over $\fri$ with error $\hat\lambda_l\leq\lambda_l+c/\sqrt{k_l-1}+\eps_l$, where $k_l=\dim\fr_l$ and $\dim\hat\fr_l=\lfloor\frac{\eps_l^2}{256\log(32/\eps_l)}k_l\rfloor$, while $\eps_l$ can be suitably chosen.\\
These results were then used in \cite{abbn} to conclude that, with a specific choice of $\eps_l$ and with $R:=\liminf_{l\to\infty}\frac{1}{l}\log k_l$, one can obtain an error $\hat\lambda_l\leq\lambda_l+c/\sqrt{k_l-1}+2^{-l\cdot R/4}$ and such that $R=\liminf_{l\to\infty}\frac{1}{l}\log\hat k_l$.\\
The implication for the $\lambda$-capacities is that the following estimates hold:
\begin{align}\label{eqn:AsgeqA}
\A_{\mathrm{s,det}}(\fri,\lambda)\geq\A_{\mathrm{det}}(\fri,\lambda)
\qquad\mathrm{and}\qquad\A_{\mathrm{s,random}}(\fri,\lambda)\geq\A_{\mathrm{random}}(\fri,\lambda).
\end{align}
The reverse statements can be proven by application of Lemma 17 in \cite{abbn}, which originally appeared in \cite{horodecki-horodecki-horodecki}: If a strong subspace transmission code for an AVQC $\fri$ with error $\lambda_l\in[0,1]$ is given, then this directly implies the existence of an entanglement transmission code with error $\hat\lambda_l=1-\frac{k_l+1}{k_l}(1-\lambda_l-\frac{1}{k_l+1})$, and of course if $\liminf_{l\to\infty}\frac{1}{l}\log k_l=R>0$ and $\limsup_{l\to\infty}\lambda_l=\lambda$ then $\limsup_{l\to\infty}\hat\lambda_l=\lambda$ and thus: If the r.h.s. below are nonzero, then
\begin{align}
\mathcal A_{\mathrm{det}}(\fri,\lambda)\geq\mathcal A_{\mathrm{s,det}}(\fri,\lambda)\qquad\mathrm{and}\qquad\mathcal A_{\mathrm{random}}(\fri,\lambda)\geq\mathcal A_{\mathrm{s,random}}(\fri,\lambda)
\end{align}
holds. But by equations (\ref{eqn:AsgeqA}), if the r.h.s. above are zero, then the l.h.s. are zero as well.
\end{proof}
\end{subsection}
\end{section}
\emph{Acknowledgements.}
This work was supported by the DFG via grant BO 1734/20-1 (H.B.) and by the BMBF via the grants 01BQ1050 and 16KIS0118 (H.B., J.N.).

\end{document}